\documentclass[journal,final,twocolumn,12pt,twoside]{IEEEtran}

\normalsize

\usepackage[utf8]{inputenc} 
\usepackage[T1]{fontenc}
\usepackage{url}
\usepackage{ifthen}
\usepackage{cite}
\usepackage[cmex10]{amsmath} 
                             
\usepackage[table]{xcolor}
\usepackage{amssymb,amsthm}
\usepackage{subcaption}
\usepackage{tikz}
\usepackage{hyperref}
\usepackage[normalem]{ulem}
\usepackage[noend]{algpseudocode}
\usepackage{algorithm}
\usepackage{tabularx,environ}
\usepackage{multirow,array}

\DeclareMathOperator{\wt}{wt}

\DeclareMathOperator{\lcm}{lcm}

\DeclareMathOperator{\rank}{rank}
\DeclareMathOperator{\Span}{span}

\DeclareMathOperator{\aut}{Aut}
\DeclareMathOperator{\ev}{ev}

\newtheorem{prop}{Proposition}
\newtheorem{defn}{Definition}

\newtheorem{cor}{Corollary}
\newtheorem{lm}{Lemma}

\newtheorem{ex}{Example}

\begin{document}
	
\title{On the Efficiency of Polar-like Decoding for Symmetric Codes}

\author{Kirill~Ivanov,~\IEEEmembership{Student Member,~IEEE,}
	R\"udiger~L.~Urbanke,~\IEEEmembership{Senior Member,~IEEE,}
	\thanks{K. Ivanov and R. L. Urbanke are with the School of Computer and Communication
		Sciences, EPFL, Switzerland e-mail: (kirill.ivanov@epfl.ch, rudiger.urbanke@epfl.ch).}
	\thanks{Part of the results of this paper is published at the ISITA 2020 \cite{ivanov2020partsym}}}

\markboth{IEEE Transactions on Communications}%
{Submitted paper}

\maketitle

\begin{abstract}
	The recently introduced polar codes constitute a breakthrough in coding theory due to their capacity-achieving property. This goes hand in hand with a quasilinear construction, encoding, and successive cancellation list decoding procedures based on the Plotkin construction. The decoding algorithm can be applied with slight modifications to Reed-Muller or eBCH codes, that both achieve the capacity of erasure channels, although the list size needed for good performance grows too fast to make the decoding practical even for moderate block lengths. 	
	
	The key ingredient for proving the capacity-achieving property of Reed-Muller and eBCH codes is their group of symmetries. It can be plugged into the concept of Plotkin decomposition to design various permutation decoding algorithms. Although such techniques allow to outperform the straightforward polar-like decoding, the complexity stays impractical.

	In this paper, we show that although invariance under a large automorphism group is valuable in a theoretical sense, it also ensures that the list size needed for good performance grows exponentially. We further establish the bounds that arise if we sacrifice some of the symmetries. Although the theoretical analysis of the list decoding algorithm remains an open problem, our result provides an insight into the factors that impact the decoding complexity.
\end{abstract}

\begin{IEEEkeywords}
	Reed-Muller codes, Polar codes, List decoding, Permutation decoding.
\end{IEEEkeywords}

\IEEEpeerreviewmaketitle

\section{Introduction}
\IEEEPARstart{R}{eed-Muller} (RM) codes and extended BCH codes are very well known in classical coding theory and have been studied for decades\cite{sloane1978ecc}. It was recently proved that these codes achieve the capacity of the binary erasure channel under ML decoding \cite{Kudekar2016RM}. The equivalent question for general channels is still open and so is low-complexity (near) maximum likelihood (ML) decoding.

The remarkable structural property of Reed-Muller codes is that they can be represented as Plotkin concatenation of two smaller Reed-Muller codes. Dumer and Shabunov \cite{dumer2006rec} proposed to perform this decomposition recursively and keep the list of most likely codewords at each stage, which brings ML performance for any rate, but the list size grows exponentially with the code length. Another idea that first appeared in \cite{dumer2006rec} and further studied in \cite{kamenev2019rm,ivanov2019permrm,geiselhart2020autrm} is to use the large automorphism group of RM codes and consider multiple codeword permutations, which improves over the list algorithm. In case of eBCH codes, the list size for near-ML decoding is significantly larger \cite{trifonov2018score}.

Polar codes \cite{Arikan2009polar} achieve the capacity of an arbitrary binary memoryless symmetric (BMS) channel. Contrary to RM codes, successive cancellation list (SCL) decoding \cite{Tal2015list} (which works similarly to Dumer-Shabunov decoder) with small list size is sufficient for near-ML performance, but the hardware implementation is challenging due to the path sorting operation \cite{giard2016hardware}. Permutation decoding works without the sorting operation and therefore provides an alternative. However, it requires a certain symmetry in the code, which is lacking for the standard constructions of polar codes as well as polar-like codes such as CRC-aided polar codes \cite{Tal2015list} and polar subcodes \cite{trifonov2016subcodes}. In \cite{kamenev2019polar} and \cite{geiselhart2021autpolar}, the authors construct polar codes, which in some cases perform better under permutation decoding than polar codes with list decoding. The construction of polar subcodes taylored for permutation decoding remains an open problem.

In this paper, we investigate the achievable properties of codes with certain symmetries. We establish the connection with the required list size for near-ML decoding and demonstrate that in case of RM and eBCH codes the automorphism group implies the exponential growth. We also investigate what happens if fewer symmetries are present and discover that a similar result holds even in this relaxed setting. A method to construct `optimal' partially symmetric monomial codes, which might be of some use for short lengths, is also presented. \footnote{The source code to reproduce the bounds and code constructions is available at https://github.com/kir1994/PartSymCodes}
\section{Background}
We use $[n]$ to denote the set $\{0,\dots,n-1\}$. $\mathbb F_q$ denotes the finite field with $q$ elements, and $\mathbb F_q^m$ is the $m$-dimensional vector space over $\mathbb F_q$. Note that we can consider the vector space $\mathbb F_q^m$ as the finite field $\mathbb F_{q^m}$ and vice versa. $\mathbb F_q^*=\mathbb F_q\setminus\{0\}$ is the multiplicative group of $\mathbb F_q$. Bold letters are used for matrices and vectors, e.g., $\mathbf{A}$ and $\mathbf b$. Given a vector $\mathbf v=(v_0,\dots,v_{m-1}),v_i\in \mathbb F_2$, we consider it as an integer $v=\sum_{i=0}^{m-1}v_i2^i$ where needed. For $\mathbf c=(c_0,\dots,c_{n-1})$, we define $\mathbf c_a^b=(c_a,\dots,c_b), 0\le a\le b<n$. $R(\mathcal C)$ denotes the rate of a binary linear code $\mathcal C$.

\subsection{Boolean functions}
Let $\{ x_0, \ldots, x_{m-1} \}$ be a collection of $m$ variables taking their values in $\mathbb F_2$, let $\mathbf v = (v_0, \ldots, v_{m-1}) \in \mathbb{F}_2^m$ be any binary $m$-tuple, and let $\wt(\cdot)$ denote the Hamming weight. Then,
\[
x^{\mathbf v} = \prod_{i=0}^{m-1} x_i^{v_i}
\]
denotes a monomial of degree $\wt(\mathbf v)$. 

A function $f(\mathbf x)=f(x_0,\dots,x_{m-1}):\ \mathbb F_2^m\to \mathbb{F}_2$ is called Boolean. Any such function can be uniquely represented as an $m$-variate polynomial:
$$
f(x_0,\dots,x_{m-1})=\sum_{\mathbf v\in \mathbb F_2^m}a_{\mathbf v}x^{\mathbf v},
$$
where $a_{\mathbf v}\in \{0,1\}$, which is also called its algebraic normal form (ANF) \cite{carlet2010boolean}. Its evaluation vector $\ev(f(\mathbf x))\in \mathbb F_2^{2^m}$ is obtained by evaluating $f$ at all points $\boldsymbol \alpha_i$ of $\mathbb{F}_2^m$. Note that any length-$2^m$ binary vector $\mathbf c$ can be considered as an evaluation vector of some function $f$. For the rest of the paper, we assume the standard bit ordering of points, i.e., $\boldsymbol\alpha_i$ being the binary expansion of integer $i$.


\subsection{Monomial and polynomial codes}
Consider a binary linear $(n=2^m,k,d)$ code $\mathcal C$ with generator matrix $\mathbf G$. Its \textit{generating set} is given by
$$
M_{\mathcal C}=\{f_i,0\le i < k | \ev(f_i)=\mathbf G_{i,*}\},
$$
where $\mathbf G_{i,*}$ are rows of $G$. The code $\mathcal C$ is called monomial if there exists $M_{\mathcal C}$ that contains only monomials (which can be obtained by applying the Gaussian elimination on the coefficients of $f_i$'s) and polynomial otherwise. The minimum distance of a monomial code can be calculated \cite[Proposition~3]{bardet2016algebraic} as 

\begin{equation}
\label{eq:mondist}
	d_{min}(\mathcal C)=2^{m-\max_{x^{\mathbf v} \in M_{\mathcal C}}\wt (\mathbf v)}.
\end{equation}

Consider the matrix $\mathbf A_m=\left(
\arraycolsep=1.15pt\def\arraystretch{0.5}
\begin{array}{cc}
	1&1\\
	0&1
\end{array} \right)^{\otimes m}$, where $\otimes m$ denotes an $m$-fold Kronecker product of a matrix
with itself. Observe that $\left(1\ 1\right)$ and $\left(0\ 1\right)$ are the evaluations over $\mathbb F_2$ of the constant monomial $1$ and the monomial $x_0$, respectively, and from the induction on $m$ it follows that the $v$-th row of $\mathbf A_m$ is an evaluation vector over $\mathbb{F}_2^m$ of the monomial $x^{\mathbf v}$. Hence, the encoding for monomial codes can be performed as $\mathbf c=\mathbf u\mathbf A_m$, where $u_v=1$ if $x^{\mathbf v}\in M_{\mathcal C}$ and $u_v=0$ otherwise.

\subsubsection{Polar codes}
A $(n=2^m,k)$ polar code \cite{Arikan2009polar} with the set of frozen symbols $\mathcal F$ is a binary linear block code generated by rows with indices $i\in[n]\setminus \mathcal F$ of the matrix $\mathbf A_m$. Our definition of polar codes is slightly different from conventional, which uses the matrix $\left(
\arraycolsep=1.15pt\def\arraystretch{0.5}
\begin{array}{cc}
	1&0\\
	1&1
\end{array} \right)^{\otimes m}$, but as already noted in \cite{bardet2016algebraic}, both definitions are equivalent and ours simplifies the polynomial notation. For given  binary memoryless symmetric (BMS) channel $W$, the set $\mathcal F$ contains the $n-k$ indices with largest bit error probability under successive cancellation decoding. Polar codes are monomial with the generating set $M_{\mathcal F}=\{x^{\mathbf i} | i \notin \mathcal F\}$.

\subsubsection{Reed-Muller codes}
A Reed-Muller code \cite{reed1954rm,muller1954rm} RM$(r,m)$ of order $r$ is spanned by the evaluations of $m$-variate monomials of degree at most $r$. The RM$(r,m)$ code has length $2^m$, dimension $\sum_{i=0}^{r}\binom{m}{i}$ and minimum distance $2^{m-r}$. RM codes are monomial with the generating set $M_{r,m}=\{x^{\mathbf v} | \wt(\mathbf v)\le r\}$.

\subsubsection{Extended BCH codes}
A $(2^m-1,k, d\ge \delta)$ primitive narrow-sense BCH code \cite{sloane1978ecc} with design distance $\delta$ has a parity check matrix $H$ with elements 
$$H_{j,i}=\alpha_i^j, 0\le i < 2^m, 0\le j < \delta-1,$$
where $\alpha_i$ are distinct elements of $\mathbb F_{2^m}^*$. An extended code is formed by adding an overall parity check symbol.

\subsection{Derivatives}
The derivative in direction $\mathbf b$ of the Boolean function $f$ is defined as
$$
(D_{\mathbf b}f)(\mathbf x) = f(\mathbf x+\mathbf b)-f(\mathbf x).
$$
In case of monomials, this expression can be written as
$$
D_{\mathbf b}x^{\mathbf v} = (\mathbf x+\mathbf b)^{\mathbf v}-x^{\mathbf v}=\prod_{i=0}^{m-1}(x_i+b_i)^{v_i}-\prod_{i=0}^{m-1}x_i^{v_i}.
$$
In case of $\wt(\mathbf b)=1$, i.e., when $\mathbf b=\mathbf e_i$, the directional derivative coincides with the partial derivative $\frac{\partial f}{\partial x_i}$.

The derivative in direction $\mathbf b$ of the code $\mathcal C$ is a binary linear code with generating set
\begin{equation}
	\label{eq:proj}
	M_{\mathcal C\to \mathbf b}=\left\{D_{\mathbf b}f_i |f_i \in M_{\mathcal C}\right\}.
\end{equation}
By definition, $D_{\mathbf b}f_i$ has identical values at coordinates $\mathbf x$ and $\mathbf x+\mathbf b$ for all $\mathbf x\in \mathbb{F}_2^m$, so we can discard the coordinates $\mathbf x+\mathbf b$ and obtain the $(n^{(\mathbf b)}=2^{m-1},k^{(\mathbf b)}=\dim M_{\mathcal C\to \mathbf b},d^{(\mathbf b)})$ code $\mathcal C^{(\mathbf b)}$.

The partial derivatives of Boolean functions are inherently connected with the Plotkin construction. Namely, any function $f$ can be decomposed as
\begin{align*}
	f(x_0,\dots,x_{m-1})&=g(x_1,\dots,x_{m-1})\\&+x_0h(x_1,\dots,x_{m-1}),
\end{align*}
where $g(x_1,\dots,x_{m-1})$ takes identical values for $x_0=0$ and $x_0=1$, whereas $x_0h(x_1,\dots,x_{m-1})$ is only nonzero when $x_0=1$.

Consider a permutation $\pi$ on $[2^m]$ and define its action on the Boolean function $f$ as another Boolean function $g$ obtained by permuting its evaluation vector, namely $\pi(f)=g:\ \ev(g)=(f(\pi(\boldsymbol\alpha_0)),\dots,f(\pi(\boldsymbol\alpha_{2^m-1})))$. The set of permutations that leave the code $\mathcal C$ invariant, i.e., map its codewords to other codewords, forms the automorphism group of a code denoted by $\aut(\mathcal C)$. Note that if $\aut(\mathcal C)$ contains the permutation $\mathbf x\to \mathbf x+\mathbf b$, then $D_{\mathbf b}\mathcal C$ is a subcode of $\mathcal C$ (assuming we don't discard the coordinates $\mathbf x+\mathbf b$).

\subsection{List and permutation decoding}
Assume that the codeword $\mathbf c$ is transmitted through a BMS channel $W$ and the received vector is $\mathbf y$. The successive cancellation (SC) algorithm performs bit-by-bit estimation of the vector $\mathbf u$ as
\begin{equation}
	\label{mSCProb}
	\tilde u_i=\begin{cases}\arg\max_{u_i\in \{0,1\}} W^{(i)}(\mathbf y_0^{n-1},\mathbf {\tilde u}_0^{i-1}|u_i), &i\notin\mathcal F,\\
		0&i\in \mathcal F,
	\end{cases}
\end{equation}
where $W^{(i)}=W^{(\{-,+\}^m)}$ and is obtained by recursive application of channel transformations 
$$W^{(-)}(y_0,y_1|u_0)=\frac{1}{2}\sum_{u_1\in \{0,1\}}W(y_0|u_0\oplus u_1)W(y_1|u_1)$$
and
$$W^{(+)}(y_0,y_1,u_0|u_1)=\frac{1}{2}W(y_0|u_0\oplus u_1)W(y_1|u_1).$$ 
We can reformulate the recursions in the decoding process as follows: 
\begin{enumerate}
	\item Recover $\mathbf c^{(-)}=\mathbf c_0^{n/2-1}\oplus \mathbf c_{n/2}^{n-1}$ transmitted through 'XOR' channel $W^{(-)}$.
	\item Recover $\mathbf c^{(+)}=\mathbf c_{n/2}^{n-1}=\mathbf{c^{(-)}}\oplus \mathbf c_0^{n/2-1}$ transmitted through 'Combine' channel $W^{(+)}$, assuming that the recovered codeword $\mathbf{\hat{c}^{(-)}}$ is correct.
	\item Return $\mathbf c=(\mathbf c^{(-)}\oplus \mathbf c^{(+)}|\mathbf c^{(+)})$
\end{enumerate}
One way to improve the finite-length performance of SC decoding is to keep track of both possible values $u_i$ so that at every step of the recursion there are at most $L$ possible vectors $\mathbf c^{(-)}$ or $\mathbf c^{(+)}$. Another approach is to observe that the code $\mathcal C^{(-)}=\{\mathbf c_0^{n/2-1}\oplus \mathbf c_{n/2}^{n-1}| \mathbf c\in \mathcal C\}$ which appears at step 1 is a partial derivative and the recursions in SC decoding follow a certain predefined ordering in which the derivatives are taken. Hence, one can take $P$ distinct orderings, perform SC decoding for each of them in parallel and return the closest codeword to the received vector.


\subsection{Decoding efficiency}
\begin{prop}
\label{prop:expl}
	Consider a code $\mathcal C$ of length $n=2^m$ and the transmission through a BMS channel $W$. If for all $i\in [m]$ holds $R(\mathcal C^{(\mathbf e_i)})>I(W^{(-)})$, SC list algorithm needs the list size $L=2^{\Omega(n)}$ to achieve ML performance.
\end{prop}
\begin{proof}
SC decoding can be considered as a two-stage process, when we first recover $\mathbf c^{(-)}\in\mathcal C^{(-)}$, assuming the transmission through the synthetic channel $W^{(-)}$, and then use it to recover $\mathbf c^{(+)}\in\mathcal C^{(+)}$, assuming the transmission through the synthetic channel $W^{(+)}$. List decoding passes $L$ candidate codewords $\mathbf c^{(-)}$ to the next stage and the decoder can only succeed if the correct codeword is in the list. If $R(\mathcal C^{(-)})>I(W^{(-)})$, we are trying to decode above the capacity and therefore need the list to be at least $2^{n(R(\mathcal C^{(-)})-I(W^{(-)}))}$ to succeed \cite[eq.~(1.6)]{shannon67lb}. It remains to recall that the code $\mathcal C^{(-)}$ is a partial derivative $\mathcal C^{(\mathbf e_i)}$, where $i$ depends on the chosen ordering.
\end{proof}
\textbf{Remark.} The condition $R(\mathcal C^{(\mathbf e_i)})>I(W^{(-)})$ in proposition \ref{prop:expl} assumes that code $\mathcal C^{(\mathbf e_i)}$ is capacity-achieving and can be safely replaced with $R(\mathcal C^{(\mathbf e_i)})>\hat I$, where $\hat I$ is the largest capacity of the channel from the same family as $W^{(-)}$ so that the ML decoding of $\mathcal C^{(\mathbf e_i)}$ almost always succeeds.

Therefore, in order to have good performance with small list size, one needs to guarantee that there is at least one derivative with sufficiently small rate. Permutation decoding is more demanding in that regard since it uses different orderings on the derivatives and needs more than one with small rate. In what follows, we demonstrate the lower bounds on this rate, and that the automorphism groups of Reed-Muller and eBCH codes ensure the exponential complexity of their list and permutation decoding.
\section{Fully symmetric codes}
\begin{defn}
	A $(2^m,k)$ code $\mathcal C$ is fully symmetric if all its partial derivatives have equal dimensions. 
\end{defn}
Let us denote this dimension as $\tilde k$. In this section, we demonstrate that for fully symmetric codes $\tilde k$ is bounded from below and in fact becomes of order $k/2$. We also show that Reed-Muller and eBCH codes are fully symmetric.

\begin{prop}
	\label{prop:flb}
	Consider a $(2^m,k)$ fully symmetric code $\mathcal C$. If its dimension can be expressed as $k=\sum_{i=0}^{l-1}\binom{m}{i}+j\frac{\lcm(l,m)}{l}$, where $0\le j\frac{\lcm(l,m)}{l} < \binom{m}{l}$, then 
	\begin{equation}
		\label{eq:flb}
		\tilde k\ge\sum_{i=0}^{l-2}\binom{m-1}{i}+j\frac{\lcm(l,m)}{m}
	\end{equation} 
\end{prop}
\noindent We get $\lcm$ in the expression due to the full symmetry constraint, which is further explained below. We call the code \textit{optimal} if it satisfies \eqref{eq:flb} with an equality.

\begin{prop}
\label{prop:fconv}
Consider a sequence of optimal fully symmetric codes $\mathcal C$ of fixed rate and increasing length $2^m$. Then for $i\in [m]$
\begin{equation}
	\label{eq:fconv}
	\lim_{m\to \infty} R(\mathcal C^{(\mathbf e_i)})=R(\mathcal C).
\end{equation} 
\end{prop}

Proposition \ref{prop:flb} implies that list or permutation decoding in any channel $W$ s.t. $\frac{\tilde k}{2^{m-1}}>I(W^{(-)})$ needs an exponential complexity to achieve the ML performance, and proposition \ref{prop:fconv} states that this condition asymptotically becomes $R(\mathcal C)>I(W^{(-)})$.

\subsection{Proof of Proposition \ref{prop:flb} for monomial codes}
\label{ss:flb}
Let us start from the rate-1 code $\mathcal C_m$. All its derivatives are also rate-1 codes and therefore $\mathcal C_m$ is fully symmetric. Any monomial code $\mathcal C$ can be constructed by removing $2^m-k$ monomials from $M_{\mathcal C_m}$ and we would like to do it in a way such that $\mathcal C$ is fully symmetric and $\tilde k$ is minimized. Observe that $\frac{\partial x^{\mathbf v}}{\partial x_i}$ is nonzero iff $v_i=1$ and hence removing $x^{\mathbf v}$ from the generating set decreases the dimensions of $\wt(\mathbf v)=\deg(x^{\mathbf v})$ partial derivatives by 1 (or equivalently, the dimensions of all derivatives are decreased on average by $\wt(\mathbf v)/m$). This implies that the optimum strategy is to remove $2^m-k$ monomials of the largest degrees.

If $k=\sum_{i=0}^{l}\binom{m}{i},0\le l \le m$, we simply remove all monomials of degree larger than $l$ and consequently each derivative contains all monomials on $m-1$ variables of degree at most $l-1$, which gives $\tilde k=\sum_{i=0}^{l-1}\binom{m-1}{i}$. Otherwise, we can write $k=\sum_{i=0}^{l-1}\binom{m}{i}+p,0<p<\binom{m}{l},$ and remove $p$ monomials of degree $l$ in addition. It follows that the dimensions of all derivatives are decreased on average by $\frac{pl}{m}$, and since the code is fully symmetric the actual decrease for any derivative should also be $\frac{pl}{m}$. Therefore, $\frac{pl}{m}$ must be an integer, which is true only if $p$ is a multiple of $\frac{\lcm(l,m)}{l}$, which implies the bound \eqref{eq:flb}.

The bound coincides with the parameters of Reed-Muller codes when $j=0$. Otherwise, the set of monomials to remove can be found by considering a bipartite graph $\mathcal G=(V_L, V_R,E)$ with left vertices $h_{i}\in V_L$ isomorphic to variables $x_i,1\le i \le m$ and right vertices $h_{v}\in V_R$ isomorphic to all degree-$l$ monomials $x^{\mathbf v}$. We draw an edge between $h_{i}$ and $h_{v}$ if the monomial $x^{\mathbf v}$ contains $x_i$. This graph is $(\binom{m-1}{l-1},l)$-biregular and we want to remove all but $j\frac{\lcm(l,m)}{l}$ of its right vertices so that the graph remains biregular, i.e., find its $(\cdot, l)$-biregular subgraph $\mathcal G'$. 

\begin{figure}
	\centering
	\scalebox{1}{\begin{tikzpicture}[thick,scale=0.55, every node/.style={scale=0.6}]
\tikzstyle{lnode}  = [radius=0.5, black, thick, fill=white];
\tikzstyle{rnodeS}  = [radius=0.5, red, thick, fill=white];
\tikzstyle{lineG}  = [-, thick,  red];

\coordinate (0) at (1,1);
\coordinate (1) at (1,3);
\coordinate (2) at (1,5);
\coordinate (3) at (1,7);
\coordinate (4) at (1,3);
\coordinate (20) at (1.5,1);
\coordinate (21) at (1.5,3);
\coordinate (22) at (1.5,5);
\coordinate (23) at (1.5,7);
\coordinate (24) at (1.5,3);

\coordinate (5) at (6,-1);
\coordinate (6) at (6,1);
\coordinate (7) at (6,3);
\coordinate (8) at (6,5);
\coordinate (9) at (6,7);
\coordinate (10) at (6,9);
\coordinate (35) at (5.5,-1);
\coordinate (36) at (5.5,1);
\coordinate (37) at (5.5,3);
\coordinate (38) at (5.5,5);
\coordinate (39) at (5.5,7);
\coordinate (310) at (5.5,9);
\draw (0) [rnodeS] circle;
\draw (0) node {$x_4$};
\draw (1) [rnodeS] circle;
\draw (1) node {$x_3$};
\draw (2) [rnodeS] circle;
\draw (2) node {$x_2$};
\draw (3) [rnodeS] circle;
\draw (3) node {$x_1$};

\draw (5) [rnodeS] circle (0.6cm);
\draw (5) node {$x_3x_4$};
\draw (6) circle (0.6cm);
\draw (6) node {$x_2x_4$};
\draw (7) circle (0.6cm);
\draw (7) node {$x_2x_3$};
\draw (8) circle (0.6cm);
\draw (8) node {$x_1x_4$};
\draw (9) circle (0.6cm);
\draw (9) node {$x_1x_3$};
\draw (10) [rnodeS] circle (0.6cm);
\draw (10) node {$x_1x_2$};

\draw  [lineG] (20) -- (35);
\draw (20) -- (36);
\draw (20) -- (38);
\draw  [lineG] (21) -- (35);
\draw (21) -- (37);
\draw (21) -- (39);
\draw (22) -- (36);
\draw (22) -- (37);
\draw  [lineG] (22) -- (310);
\draw (23) -- (38);
\draw (23) -- (39);
\draw  [lineG] (23) -- (310);

\end{tikzpicture}}
	\caption{$(3,2)$-regular $\mathcal G$ and $(1,2)$-regular $\mathcal G'$}
	\label{fig:subgraph}
\end{figure}
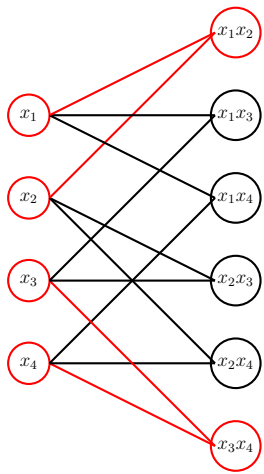

Figure \ref{fig:subgraph} demonstrates an example of graph $\mathcal G$ for $m=4,l=2$ and one of its possible $(1,2)$-regular subgraphs $\mathcal G'$ (in red). Such $\mathcal G'$ can be found as a maximum flow solution for the network with the source connected to all left vertices with capacity-$\frac{\lcm(l,m)}{l}$ edges, the sink connected to all right vertices with capacity-$m$ edges and all $e\in E$ having the unit capacity.

\subsection{Proof of Proposition \ref{prop:flb} for polynomial codes}
\label{ssFullPoly}
We start from a simple case and consider an $(2^m,k)$ code $\mathcal C_r\subseteq \Span \{x^{\mathbf v}|\wt(\mathbf v)=r\}$ for some fixed $r$. By definition, $0\le k\le \binom{m}{r}$, and consequently $0\le \tilde k\le \binom{m-1}{r-1}$. Code $\mathcal C_r$ by definition is spanned by $k$ degree-$r$ linearly independent homogeneous polynomials $f_s, 1\le s\le k$. Assume now a certain ordering on monomials $x^{\mathbf v^{(j)}}$, e.g., lexicographic w.r.t. $\mathbf v^{(j)}$, and consider the $k\times \binom{m}{r}$ matrix $\mathbf M$ such that $M_{s,j}=1$ if $f_s$ includes $\mathbf v^{(j)}$. $\mathbf M$ is a basis of the linear space of all polynomials whose evaluations are codewords of $\mathcal C_r$ and therefore has full row rank, so we can use Gaussian elimination to transform it into $\mathbf {\tilde M}=\left(\mathbf I \ *\right)\mathbf P$, where $\mathbf P$ is a column permutation matrix. Let us further define a vector $\boldsymbol \phi$ s.t. $\phi_j=1$ if $j$-th column of $\mathbf {\tilde M}$ is a column of the identity matrix.

Similarly, the generator of the linear space corresponding to the partial derivative $\frac{\partial}{\partial x_i}$ is a matrix $\mathbf {\tilde M}_i$ obtained by removing all columns of $\mathbf {\tilde M}$ but the ones that correspond to monomials that include $x_i=1$. Its dimension is equal to $\rank \mathbf {\tilde M}_i$. From the construction it follows that $\rank \mathbf {\tilde M}_i\ge |\{j| \phi_j=1 \wedge v^{(j)}_i=1\}|$. In case of monomial codes we have $\rank \mathbf {\tilde M}_i=|\{j| \phi_j=1 \wedge v^{(j)}_i=1\}|$, so it only remains to see that the bound \eqref{eq:flb} minimizes $\max_i |\{j| \phi_j=1 \wedge v^{(j)}_i=1\}|$. The extension to the general case is straightforward.

\begin{ex}
	Consider $m=4,r=2,k=4$ and assume $\boldsymbol \phi=(0,1,1,1,1,0)$, where the ordering on degree-2 monomials is $(x_1x_2,x_1x_3,x_1x_4,x_2x_3,x_2x_4,x_3x_4)$. This vector corresponds to the matrix
	$$
	\mathbf {\tilde M}=\begin{pmatrix}
		*&1&0&0&0&*\\
		*&0&1&0&0&*\\
		*&0&0&1&0&*\\
		*&0&0&0&1&*\\
	\end{pmatrix},	
	$$
	where $*$ can be any binary value. The partial derivatives correspond to the matrices 
	$$
	\mathbf {\tilde M}_1=
	\begin{pmatrix}
		*&1&0\\
		*&0&1\\
		*&0&0\\
		*&0&0\\
	\end{pmatrix}	
	\mathbf {\tilde M}_2=\begin{pmatrix}
		*&0&0\\
		*&0&0\\
		*&1&0\\
		*&0&1\\
	\end{pmatrix}\\,
	$$	
	$$
	\mathbf {\tilde M}_3=\begin{pmatrix}
	1&0&*\\
	0&0&*\\
	0&1&*\\
	0&0&*\\
	\end{pmatrix}	
	\mathbf {\tilde M}_4=\begin{pmatrix}
		0&0&*\\
		1&0&*\\
		0&0&*\\
		0&1&*\\
	\end{pmatrix}.
	$$	
	We have $\max_i \rank \mathbf {\tilde M}_i\ge2$ and we know that the code spanned by $\{x_1x_3,x_1x_4,x_2x_3,x_2x_4\}$ has $\max_i \rank \mathbf {\tilde M}_i=2$.
\end{ex}

\subsection{Proof of Proposition \ref{prop:fconv}}
\label{ss:fAs}
Let $m$ be an odd number and consider an optimal fully symmetric code of rate $1/2$. Its dimension can be expressed as $k=2^{m-1}=\sum_{i=0}^{\lfloor m/2\rfloor}\binom{m}{i}$ and its derivatives have dimension
$$
\tilde k=\sum_{i=0}^{\lfloor m/2\rfloor-1}\binom{m-1}{i}.
$$
Now consider $\left|\frac{\tilde k}{2^{m-1}}-\frac12\right|=\left|\frac{\binom{m-1}{\lfloor m/2\rfloor}}{2^m}\right|$, which goes to 0 with $m\to \infty$, and to finish the proof it remains to notice that the bound \eqref{eq:flb} is convex, which gives the same convergence for all values of $k$. Assume now that $\frac{\partial C}{\partial x_i}$ is a subcode of $\mathcal C$, i.e., $\aut(\mathcal C)$ contains the permutation $\mathbf x\to \mathbf x+\mathbf e_i$. Since any permutation $\mathbf x\to \mathbf x+\mathbf b$ can be decomposed into a product of transpositions, we have $\dim \frac{\partial C}{\partial x_i}\le k/2$ due to Legeay \cite[Corollary 1]{legeay2011perm} and therefore the lower bound on the derivative code rate converges to the worst-case value.

\subsection{Symmetry of RM and eBCH codes}
\label{ss:affsym}

\begin{prop}
	Reed-Muller codes are fully symmetric.
\end{prop}
\begin{proof}
	Indeed, by construction $M_{r,m}$ includes all $m$-variate monomials up to degree $r$ and the generating set of any partial derivative consists of all $(m-1)$-variate monomials up to degree $r-1$.	
\end{proof}

\begin{lm}
\label{lm:permeq}
If $\aut(\mathcal C)$ contains permutation $\pi$ s.t. $\pi(\mathbf x+\mathbf {\tilde b})=\pi(\mathbf x)+\mathbf b$ for some nonzero $\mathbf b,\mathbf {\tilde b}\in \mathbb F_2^m$, then the codes induced by the direction derivatives $D_{\mathbf b}$ and $D_{\mathbf {\tilde b}}$ are permutation equivalent.
\end{lm}

\begin{proof}
	Consider some function $f(\mathbf x)$ which is a codeword of $\mathcal C$ and its permutation $g(\mathbf x)=f(\pi(\mathbf x))$. Take the derivatives in directions $\mathbf b$ and $\mathbf {\tilde b}$:
	\begin{align*}
		(D_{\mathbf b}f)(\mathbf x)&=f(\mathbf x)+f(\mathbf x+\mathbf b)\\
		(D_{\mathbf {\tilde b}}g)(\mathbf x)&=g(\mathbf x)+g(\mathbf x+\mathbf {\tilde b})\\&=f(\pi(\mathbf x))+f(\pi(\mathbf x+\mathbf {\tilde b}))
		\\&=f(\pi(\mathbf x))+f(\pi(\mathbf x)+\mathbf b)
	\end{align*}
	It follows that $D_{\mathbf {\tilde b}}g$ can be obtained from $D_{\mathbf b}f$ by map $\mathbf x\to \pi(\mathbf x)$. Since both $f$ and $g$ are codewords of $\mathcal C$, we can conclude that any codeword of $D_{\mathbf {\tilde b}} \mathcal C$ can be obtained from a codeword of $D_{\mathbf b} \mathcal C$ by permutation and vice versa and consequently the derivatives of $\mathcal C$ in directions $\mathbf b$ and $\mathbf {\tilde b}$ lead to permutation equivalent codes.
\end{proof}
\begin{prop}
	\label{prop:permeq}
	Affine-invariant codes are fully symmetric.
\end{prop}
\begin{proof}
A code is affine-invariant if $\aut(\mathcal C)$ contains permutations $x\to ax+b$ for $a\in \mathbb F_{2^m}^*,b\in  \mathbb F_{2^m}$, where the multiplication $ax$ is performed in the finite field $\mathbb F_{2^m}$. A permutation $x\to ax$ satisfies the conditions of Lemma \ref{lm:permeq} for all pairs $(b, \tilde b)$ s.t. $b\tilde b^{-1}=a$ and it follows that all derivative codes are permutation equivalent.
\end{proof}
\begin{cor}
	eBCH codes are fully symmetric.
\end{cor}
\begin{proof}
Indeed, eBCH codes are affine-invariant \cite[Theorem 16]{sloane1978ecc} and therefore fully symmetric.
\end{proof}


\begin{figure}
	\centering
	\includegraphics[width=\linewidth]{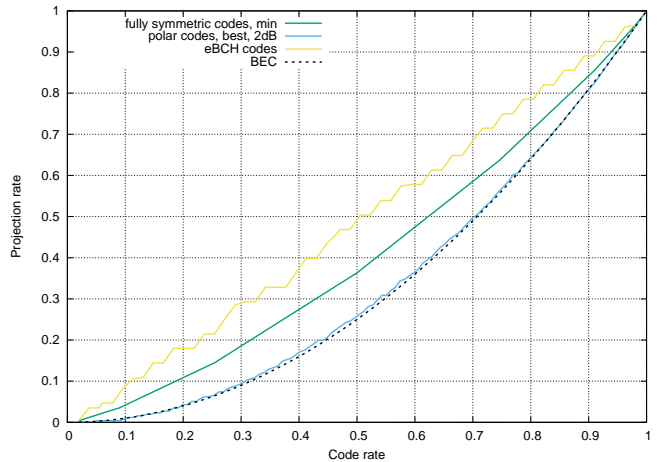}
	\caption{Bounds on the derivative code rates for fully symmetric codes, $n=512$}
	\label{fig:rates9}
\end{figure}

Figure \ref{fig:rates9} shows the actual derivative code rates for eBCH codes of length $512$ along with the lower bound \eqref{eq:flb} on the derivative rates for fully symmetric codes (recall that RM codes achieve this bound) compared to the smallest derivative rate for polar codes constructed using the Gaussian approximation \cite{trifonov2012efficient} for $E_b/N_0=2$dB. The capacity transformation for BEC is given as a reference (if $I(W)=(1-\varepsilon)$, $I(W^{(-)})=(1-\varepsilon)^2$). Observe that the bound is rather loose for eBCH codes, which are close to an upper bound $k/2$. Potentially a better bound might be derived by taking more structural properties into consideration rather than only full symmetry, which we leave as a direction for the future research. However, this plot provides a good demonstration why the list size for near-ML decoding of eBCH codes grows even faster than for RM codes (and why for both codes it quickly becomes impractical). An interesting observation is that the smallest derivative rate for polar codes is close to the capacity of the 'XOR' erasure channel despite being constructed for the Gaussian channel.

\section{Partially symmetric codes}

We demonstrated that the full symmetry puts a rather restrictive lower bound on the dimensions of the derivatives. In this section, we show what happens if we demand fewer symmetries and derive the equivalent bounds.

\begin{defn}
A $(2^m,k)$ code $\mathcal C$ is $t$-symmetric if $t$ of its partial derivatives have equal dimensions, which we denote as $\tilde k_t$, and $m-t$ have dimensions strictly greater.
\end{defn}
In other words, there exists a set of \textit{target derivatives} $\mathcal H_t,|\mathcal H_t|=t$, such that $\forall \mathbf e_i\in \mathcal H_t\ \dim \mathcal C^{(\mathbf e_i)}=\tilde k_{t}$ and $\forall \mathbf e_i\notin \mathcal H_t\ \dim \mathcal C^{(\mathbf e_i)}>\tilde k_{t}$. A code is fully symmetric if $t=m$, non-symmetric if $t=1$ and partially symmetric otherwise. Reed-Muller codes are fully symmetric and polar codes are in general non-symmetric. Without loss of generality, we assume $\mathcal H_t=\{\mathbf e_i|i\in [t]\}$.

\begin{prop}
	\label{prop:lb}
	If a $t$-symmetric code $\mathcal C$ has dimension $k=\sum_{i=0}^{l-1}\binom{t}{i}2^{m-t}+j\frac{\lcm(l,t)}{l}$, then 
	\begin{equation}
		\label{eq:lb}
		\tilde k_t\ge \sum_{i=0}^{l-2}\binom{t-1}{i}2^{m-t}+j\frac{\lcm(l,t)}{t}.
	\end{equation} 
\end{prop}
\noindent We get $\lcm$ in the expression due to the $t$-symmetry constraint, which is further explained below. Again we call a partially symmetric code \textit{optimal} if it satisfies \eqref{eq:lb} with an equality.
\begin{prop}
\label{prop:conv}
Consider a sequence of optimal $t$-symmetric codes $\mathcal C$ of fixed rate and increasing length $2^m$, where $t$ is an increasing function of $m$. Then for $i\in [t]$
\begin{equation}
	\label{eq:conv}
	\lim_{m\to \infty} R(\mathcal C^{(\mathbf e_i)})=R(\mathcal C).
\end{equation} 
\end{prop}
Therefore, even in more relaxed setting we get the lower bound similar to the one for fully symmetric code and identical asymptotic result. 

\subsection{Proof of Proposition \ref{prop:lb}}
Define $\tau_{\mathbf v}=|\{i\in [t]|v_i=1\}|$, i.e. the number of variables $\{x_0,\dots,x_{t-1}\}$ in the monomial $x^{\mathbf v}$. We start from rate-1 code $\mathcal C_m$ and using the same argument as in section \ref{ss:flb} we conclude that removing $2^m-k$ monomials of the largest $\tau_{\mathbf v}$ gives the optimal $t$-symmetric code. The number of monomials $x^{\mathbf v}$ s.t. $\tau_{\mathbf v}=l$ is $\binom{t}{l}2^{m-t}$ since $l$ of the variables $\{x_0,\dots,x_{t-1}\}$ can be selected in $\binom{t}{l}$ ways with any combination of the remaining $m-t$.

Figure \ref{fig:rates9tsym} demonstrates the lower bound \ref{eq:lb} on the derivative rates of partially symmetric codes for $t>2$ and $n=512$. In case of 3-symmetric codes, the bound is close to the BEC curve, which is similar to the best derivative for polar codes, so we can expect rather good list decoding performance. However, it quickly grows with $t$, so we expect the large list size for near-ML decoding except for the low- and high-rate regions.

\begin{figure}
	\centering
	\includegraphics[width=\linewidth]{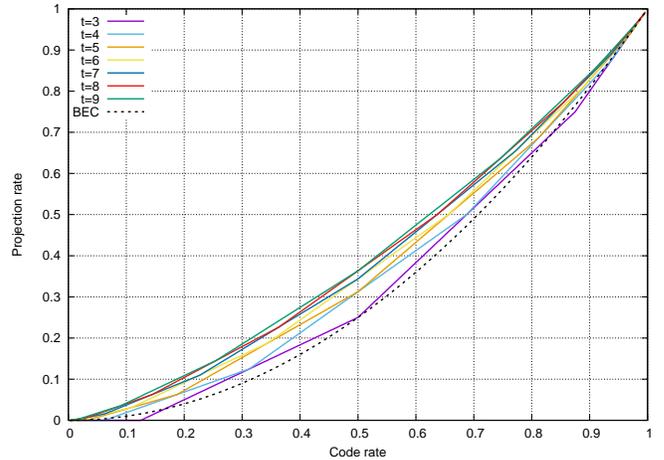}
	\caption{Lower bounds on the derivative code rates for partially symmetric codes, $n=512$}
	\label{fig:rates9tsym}
\end{figure}


Similarly to the section \ref{ssFullPoly}, the bound for monomial codes also holds for polynomial codes.

\subsection{Proof of Proposition \ref{prop:conv}}
Let $t$ be an odd number and consider an optimal $t$-symmetric code of rate $1/2$. Its dimension can be expressed as $k=2^{m-1}=2^{m-t}\sum_{i=0}^{\lfloor t/2\rfloor}\binom{t}{i}$ and its target derivatives have dimension
$$
\tilde k_t=2^{m-t}\sum_{i=0}^{\lfloor t/2\rfloor-1}\binom{t-1}{i}.
$$
If $t$ is an increasing function of $m$, similarly to section \ref{ss:fAs} the expression $\tilde k_t/2^{m-1}$ converges to $\frac12$ and due to convexity of bound \eqref{eq:lb}, the same holds for all values of $k$. Therefore, any sequence of $(2^m,k)$ binary linear codes with the partial symmetry growing with $m$ is asymptotically bad for polar-like decoding.

\subsection{Code construction}
\label{ssDesign}
The construction of the optimal $t$-symmetric monomial codes, i.e., achieving \eqref{eq:lb} with equality, is summarized in Algorithm \ref{alg:optcode}.
\begin{algorithm}
	\caption{Construction of optimal partially symmetric code}
	\label{alg:optcode}
\begin{algorithmic}
	\State Take $M_{\mathcal C}=M_{\mathcal C_m}$, $k'=|M_{\mathcal C_m}|=2^m$
	\State $\hat l\gets t$
	\While{$k'-2^{m-t}\binom{t}{\hat l} \ge k$}
	\State Remove from $M_{\mathcal C}$ all monomials with $\tau_{\mathbf v}=\hat l$
	\State $\hat l\gets \hat l-1$, $k'\gets k'-2^{m-t}\binom{t}{\hat l}$
	\EndWhile
	\State $\hat d\gets m-t+\hat l$	
	\While{$k'-\binom{t}{\hat l}\binom{m-t}{\hat d-\hat l} \ge k$}
	\State Remove from $M_{\mathcal C}$ all degree-$\hat d$ monomials with $\tau_{\mathbf v}=\hat l$
	\State $\hat d\gets \hat d-1$, $k'\gets k'-\binom{t}{\hat l}\binom{m-t}{\hat d-\hat l}$
	\EndWhile
	\While{$k'-\binom{t}{\hat l}\ge k$}
	\State Pick a degree-$(\hat d -\hat l)$ monomial $x^{\mathbf s}$ s.t. $\tau_{\mathbf s}=0$
	\State Remove from $M_{\mathcal C}$ all degree-$\hat d$ monomials with $\tau_{\mathbf v}=\hat l$ that contain $x^{\mathbf s}$
	\State $k'\gets k'-\binom{t}{\hat l}$
	\EndWhile
	\Return $M_{\mathcal C}$
\end{algorithmic}
\end{algorithm}
Note that for the small values of $t$ this construction leads to poor minimum distance. For example, in case of $2$-symmetric codes with dimension $k\ge 2^m-2^{m-2}$ only at most $2^{m-2}$ monomials with $\tau_{\mathbf v}=2$ are removed. Consequently, the generating set contains a degree-$(m-1)$ monomial $\mathbf x^{\mathbf {\hat v}}$ with $\tau_{\mathbf {\hat v}}=1$ and therefore from \eqref{eq:mondist} the code has minimum distance at most 2. 

In practice, one can construct $t$-symmetric codes as subcodes of some Reed-Muller codes RM$(r,m)$ to guarantee that the minimum distance is at least $2^{m-r}$. In this case, at step 1 we start from $M_{r,m}$ instead of $M_{\mathcal C_m}$, at step 2 the term $2^{m-t}$ is replaced with $\binom{t}{\hat l}\sum_{i=0}^{\min(m-t,r-\hat l)} \binom{m-t}{i}$ and at step 3 the initial value of $\hat d$ becomes $\min(m-t+\hat l,r)$ (since after step 1 all monomials with the degree greater than $r$ are already removed and therefore out of consideration).

\begin{ex}
	Consider $m=4,t=3$ and $k=8$. All monomials with nonzero $\tau_{\mathbf v}$ are listed in Table \ref{tab:monrem} sorted in the removal order.
	\begin{table}[]
		\caption{Monomials to remove.}
		\label{tab:monrem}
		\centering
		\begin{tabular}{|c|c|c|}
			\hline
			\textbf{$l$} &\textbf{Impact on dimension} &\textbf{Monomials} \\ \hline
			$3$ & Remove 1 monomial  & $x_1x_2x_3x_4$\\\cline{3-3}
			& $\tilde k$ decreases by 1 & $x_1x_2x_3$    \\ \hline
			$2$   & Remove 3 monomials  & $x_1x_2x_4$, $x_1x_3x_4$, $x_2x_3x_4$   \\ \cline{3-3}
			& $\tilde k$ decreases by 2&  $x_1x_2$, $x_1x_3$, $x_2x_3$\\ \hline
			1 & Remove 3 monomials & $x_1x_4$, $x_2x_4$, $x_3x_4$\\ \cline{3-3}
			& $\tilde k$ decreases by 1 & $x_1$, $x_2$, $x_3$\\ \hline
		\end{tabular}
	\end{table}	
	Start from $M_{4}, k'=16$ and go to step 2. Set $\hat l=3$. $16-2^{4-3}\binom{3}{3}=14\ge 8$, so we remove all monomials that contain $x_1x_2x_3$ ($x_1x_2x_3x_4$ and $x_1x_2x_3$), now $k'=14$ and $\hat l=2$. $14-2^{4-3}\binom{3}{2}=8\ge 8$, so we remove all monomials that contain $x_1x_2$, $x_1x_3$ or $x_2x_3$  ($x_1x_2x_4$, $x_1x_3x_4$, $x_2x_3x_4$ and $x_1x_2$, $x_1x_3$, $x_2x_3$), now $k'=8$ and the construction procedure is terminated since $k'=k$. 
	
	The constructed $(16,8,4)$ code has generating set $M_{\mathcal C_{4,3}}=\{x_1x_4, x_2x_4, x_3x_4,x_1,x_2,x_3,x_4,1\}$ and all of its target derivatives have the generating set $\{x_4,1\}$ of cardinality $2$.
\end{ex}
Assume that $k$ satisfies proposition \ref{prop:lb} and Algorithm \ref{alg:optcode} ends with $k'>k$. This means that for some fixed degree-$(\hat d -\hat l)$ monomial $x^{\mathbf s},\tau_{\mathbf s}=0$ we need to remove $k'-k$ degree-$\hat d$ monomials with $\tau_{\mathbf v}=\hat l$ that contain $x^{\mathbf s}$ so that the dimensions of all target derivatives are decreased by $\frac{(k'-k)\hat l}{t}$. The set of monomials to remove can be found using the same bipartite graph formulation as in section \ref{ss:flb}.

\subsection{Performance of partially symmetric monomial codes}
We consider the transmission via additive white Gaussian noise (AWGN) channel with binary phase shift keying (BPSK) modulation. We compare the list and permutation decoding of optimal $(256,127)$ 3-symmetric and $(256,128)$ 5-symmetric monomial codes that are subcodes of RM$(4,8)$, which are constructed with the proposed algorithm. The set of permutations $\pi_i$ is selected as in \cite{kamenev2019polar}, namely by sorting all $m!$ factor graph layer permutations by the SC decoding error probability and picking $P$ smallest such that the Hamming distance between any pair $(\pi_{i'},\pi_{i''})$ is at least 5 so that they are more likely to correct different error patterns. We observed that this method performs better than randomly choosing from $t!$ layer permutations. Maximum likelihood performance is estimated by choosing sufficiently large list size so that the returned codeword is always not farther from the received vector than the true one. For each data point, the simulation is stopped after reaching 1000 decoding error events or 1000000 transmitted frames.
\begin{figure}
	\centering
	\includegraphics[width=\linewidth]{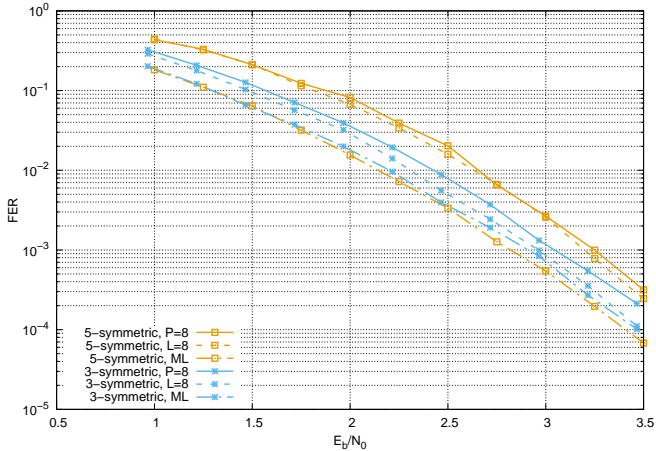}
	\caption{List and permutation decoding performance,  $(256,127)$ 3-symmetric and $(256,128)$ 5-symmetric monomial codes}
	\label{fig:fer8perm}
\end{figure}

The results are presented at figure \ref{fig:fer8perm}. Despite almost identical ML performance, 3-symmetric codes perform better under SCL decoding. In case of 5-symmetric codes, permutation decoding is as efficient as SCL. A similar behavior for partially symmetric codes is also observed in \cite{geiselhart2021autpolar} and \cite{pillet2021polar}, where a larger group of permutations is used for the decoding.
\section{Discussion}
Our results show that even the partial symmetry often implies an exponential growth of list size. The case of permutation decoding is much trickier to analyze. However, we know that it needs a rather large automorphism group to work efficiently, which in turn increases the partial symmetry and therefore the SC error probability. We conjecture that for larger codes it quickly becomes much less efficient, although the limitations of permutation decoding for short and moderate block lengths are yet to be investigated.

Another question is the construction of partially symmetric polynomial codes. Polar codes with CRC \cite{Tal2015list} and dynamic frozen symbols \cite{trifonov2017rand} demonstrate a significant performance improvement under list decoding despite being rather simple to construct, so we would like to have something similar for the permutation decoding. Experiments show that applying CRC does not give the same effect for the permutation decoding as it does for the list decoding. The design of codes with dynamic frozen symbols is also highly nontrivial. One possible approach might be to start from an optimal monomial code with further refinement of the corresponding linear subspace. However, we did not succeed in constructing the codes of practical interest with this method. 

\section{Conclusion}
In this paper, we studied how permutation group of the code impacts its list decoding performance. One result is that highly symmetric codes such as Reed-Muller or eBCH codes need exponentially growing list size and hence variations of SC algorithm quickly become inefficient for larger code lengths. We also show that the similar result holds in the case of much lesser symmetry.

\section*{Acknowledgements}
The authors thank the anonymous reviewers for many helpful comments, which have greatly improved the quality of the paper.
The authors thank Henry Pfister for the useful feedback and pointing out that Lemma \ref{lm:permeq} can be formulated for more general class of permutations rather than $x\to ax$.


\bibliographystyle{IEEEtranTCOM}
\bibliography{template}
\end{document}